\documentclass[12pt,english,reqno]{amsart}
\usepackage{amsmath,amssymb,amsthm,amsfonts,url,hyperref,cite}
\usepackage[T1]{fontenc}
\usepackage[cp1251]{inputenc}
\usepackage[active]{srcltx}
\usepackage{textcomp}
\newtheorem{thm}{Theorem}

\newtheorem{prop}{Proposition}
\theoremstyle{definition}
\newtheorem{defn}{Definition}
\theoremstyle{remark}
\newtheorem{rem}{Remark}

\title[Hyperbolic quasiperiodic solutions]
{Existence theorem\\ for hyperbolic quasiperiodic solutions\\ of Lagrangian systems
on Riemannian manifolds}
\author[I. Parasyuk]
{Igor Parasyuk}
\address{Taras Shevchenko National  Univesity of
Kyiv\\ Volodymyrs'ka 64, Kyiv, 01033, Ukraine}
\email{pio@univ.kiev.ua}

\begin{document}

 \begin{abstract}We establish new sufficient conditions for the existence of classical
hyperbolic quasiperiodic solutions for natural Lagrangian system on Riemannian manifold
with time-quasiperiodic force function.\end{abstract}

\maketitle

\textbf{1. Introduction.} On Riemannian manifold $(\mathcal{M},\left\langle \cdot,\cdot\right\rangle )$
with Riemannian metric $\left\langle \cdot,\cdot\right\rangle $ and
Levi-Civita connection $\nabla$, consider a natural Lagrangian system
with Lagrangian density
\begin{gather}
L(t\omega,x,\dot{x})=\mathrm{K}(\dot{x})+W(t\omega,x).\label{eq:quat_Lagr}
\end{gather}
where $\mathrm{K}(\dot{x}):=\left\langle \dot{x},\dot{x}\right\rangle /2$
is the kinetic energy, $W(\cdot,\cdot)\text{:}\mathbb{T}^{k}\times\mathcal{M}\mapsto\mathbb{R}$
is a smooth force function, $\mathbb{T}^{k}=\mathbb{R}^{k}/(2\pi\mathbb{Z}^{k})$
is $k$-dimensional torus with natural angle coordinates $\varphi=(\varphi_{1},\ldots,\varphi_{k})|\bmod2\pi$,
and $\omega\in\mathbb{R}^{k}$ is a basis frequency vector with rationally
independent components.

In \cite{ParRus12}, there was developed a variational method for
detection of weak quasiperiodic solutions to system
with Lagrangian density \eqref{eq:quat_Lagr}. Earlier this method was applied for
analogous purposes in the case of Lagrangian systems in Euclidean
space \cite{Blo88,Blo89,Blo89_a,BerZha96,Maw98,ZakPar99_0,AyaBlo08,Kua12}, as well as for Lagrangian systems on Riemannian manifold
with nonpositive sectional curvature ~\cite{ZahPar99,ZahPar99_a}. In contrary
to the last two papers, the authors of \cite{ParRus12} managed to avoid the requirement
concerning nonpositiveness of sectional curvature, but instead they
impose additional conditions in terms of an auxiliary function. To
formulate these conditions we shall use the following notations: $\nabla_{\xi}$,
$\left\Vert \xi\right\Vert $, $\nabla V(x)$ and $H_{V}(x)$ are,
respectively, the covariant derivative along a tangent vector $\xi$,
the norm of this vector, the gradient of a smooth function $V(\cdot):\mathcal{M}\mapsto\mathbb{R}$
at a point $x$, and the Hessian at this point.

The above conditions are as follows.
There exists a bounded smooth
function $V(\cdot)\text{:}\mathcal{M}\mapsto\mathbb{R}$ satisfying
the following conditions:

\textbf{C1}: the set $\mathcal{D}:=\left\{ x\in\mathcal{M}:2\lambda_{V}(x)+\left\Vert \nabla V(x)\right\Vert ^{2}>0\right\} ,$
where
\[
\lambda_{V}(x):=\min\left\{ \left\langle H_{V}(x)\xi,\xi\right\rangle :\left\Vert \xi\right\Vert =1,\;\xi\in T_{x}\mathcal{M}\right\}
\]
stands for minimal eigenvalue of Hessian $H_{V}(x)$, is nonempty
and, for a noncritical value $v\in V(\mathcal{D})$, there exists
a bounded connected component $\Omega$ of sublevel set $V^{-1}\left(-\infty,v\right)$
such that $\bar{\Omega}:=\Omega\cup\partial\Omega\subseteq\mathcal{D}$;

\textbf{C2}:\textbf{ }the restriction of quadratic form $\left\langle H_{V}(x)\xi,\xi\right\rangle $
on $T_{x}\partial\Omega$ is positive definite for all $x\in\partial\Omega$,
and
\begin{gather}
\min_{x\in\bar{\Omega}}\left\{ \mu_{V}(x)-2K^{\ast}(x)\right\} >0,\label{eq:conv_mainineq}
\end{gather}
where
\begin{gather*}
\mu_{V}(x):=\min\left\{ \left\langle H_{V}(x)\xi,\xi\right\rangle -\frac{1}{2}\left\langle \nabla V(x),\xi\right\rangle ^{2}:\left\Vert \xi\right\Vert =1,\xi\in T_{x}\mathcal{M}\right\} ,
\end{gather*}
and $K^{\ast}(x)$ is maximal sectional (Riemannian) curvature over
two-dimensional tangent planes of tangent space $T_{x}\mathcal{M}$
(see \cite[Sect. 3.6]{GKM71} for the definition);

\textbf{C3}: any two points $x,y\in\Omega$ can be connected in $\bar{\mathcal{D}}$
by a geodesic segment of conformally equivalent Riemannian metric
$\left\langle \cdot,\cdot\right\rangle _{V}:=\mathrm{e}^{V}\left\langle \cdot,\cdot\right\rangle $.

Unfortunately,
Condition \textbf{C3} is of ``noncoefficient'' and nonlocal
character and as a result its verification may be quite difficult.
In the present paper we shall show that actually the verification of this
condition is needless, since it is fulfilled automatically. This fact allows as
to assert that Conditions\textbf{ C1}, \textbf{C2} together with certain
additional properties of function $W(\cdot,\cdot)$ guarantee
the existence of function $u(\cdot):\mathbb{T}^{k}\mapsto\mathcal{M}$
which is associated with a week Besicovitch quasiperiodic solution $x(t):=u(t\omega)$
to system with Lagrangian density $L(t\omega,x,\dot{x})$ (Theorem~\ref{thm:quat_1}).
This solution has the following extremal property: it is a limit of
sequence minimizing the functional
\begin{gather}
\mathcal{L}[x(\cdot)]:=\lim_{T\to\infty}\frac{1}{2T}\intop_{-T}^{T}L(t\omega,x(t),\dot{x}(t))\mathrm{d}t \label{eq:quat_Lagr_lim}
\end{gather}
on a metric space formed as a closure of the set of smooth uniformly
quasiperiodic functions $x(\cdot):\mathbb{R}\mapsto\Omega$ by the
pseudometric
\begin{gather}
d_{1}(x_{1}(\cdot),x_{2}(\cdot)):=\lim_{T\to\infty}\frac{1}{2T}\intop_{-T}^{T}\left[\left\Vert \dot{x}_{1}(t)-x_{2}(t)\right\Vert ^{2}+\rho^{2}(x_{1}(t),x_{2}(t))\right]\mathrm{d}t.\label{eq:quat_metric_d1}
\end{gather}
where $\rho(\cdot,\cdot):\mathcal{M}\times \mathcal{M}\mapsto \mathbb{R}$ is the distance function on  $(\mathcal{M},\left\langle \cdot,\cdot\right\rangle )$.
 Basing on the approach proposed
in~\cite{BerZha96} and developed in \cite{BloPen01}, we shall prove that
$x(\cdot)$ is classical uniformly quasiperiodic solution of the system
with Lagrangian density $L(t\omega,x,\dot{x})$ (Theorem~\ref{thm:quat_class_QPS}).
Besides, we establish that the system in variations along $x(t)$
is exponentially dichotomic on  whole real axis (Theorem~\ref{thm:quat_ExpDych}).

\medskip{}

\textbf{2. Existence theorem for weak quasiperiodic solution to Lagrangian
system on Riemannian manifold}.

By the Nash theorem~\cite{Nash56}, the manifold $\left(\mathcal{M},\left\langle \cdot,\cdot\right\rangle \right)$
admits a smooth isometric embedding $\iota:\mathcal{M}\mapsto\mathbb{E}^{n}$
into Euclidean space $\mathbb{E}^{n}=\left(\mathbb{R}^{n},\left(\cdot,\cdot\right)\right)$
of a dimension $n>m$, so that $\left\langle \xi,\eta\right\rangle =\left(\iota_{\ast}\xi,\iota_{\ast}\eta\right)$
for any $\xi,\eta\in T_{x}\mathcal{M}$. Here $\left(\cdot,\cdot\right)$
is standard dot product in $\mathbb{R}^{n}$, $\iota_{\ast}$ stands
for derivative of the embedding map $\iota$.

Denote by $D_{\omega}f(\varphi):=\frac{\mathrm{d}}{\mathrm{d}t}\bigl|_{t=0}f(\varphi+t\omega)$
the derivative in direction of vector $\omega$ for a function $f(\cdot):\mathbb{T}^{k}\mapsto\mathbb{E}^{n}$
and by $\nabla W(\varphi,x)$ the gradient of the function $W(\varphi,\cdot):\mathcal{M}\mapsto\mathbb{R}$
at point $x$ for fixed $\varphi\in\mathbb{T}^{k}$.
\begin{prop}
\label{prop:quat_0_1}A function $t\mapsto x(t):=u(t\omega)$ of class
$\mathrm{C}^{2}\!\left(\mathbb{R};\mathcal{M}\right)$ where $u(\cdot)\in\mathrm{C}\!\left(\mathbb{T}^{k}\!\mapsto\!\mathcal{M}\right)$
and $D_{\omega}^{2}u(\cdot)\in\mathrm{C}\!\left(\mathbb{T}^{k}\!\mapsto\!\mathcal{M}\right)$
is a classical quasiperiodic solution of Lagrangian system on $\left(\mathcal{M},\left\langle \cdot,\cdot\right\rangle \right)$
with Lagrangian density \eqref{eq:quat_Lagr} iff for any mapping
$h(\cdot)\in\mathrm{C}\!\left(\mathbb{T}^{k};T\mathcal{M}\right)$
such that $D_{\omega}\iota_{\ast}h(\cdot)\in\mathrm{C}\!\left(\mathbb{T}^{k};\mathbb{E}^{n}\right)$
and $h(\varphi)\in T_{u(\varphi)}\mathcal{M}$ at every point $\varphi\in\mathbb{T}^{k}$
the function $u(\cdot)$ satisfies the equality
\begin{gather}
\intop_{\mathbb{T}^{k}}\left[\left(D_{\omega}\iota\circ u(\varphi),D_{\omega}\iota_{\ast}h(\varphi)\right)+\left(\iota_{\ast}\nabla W(\varphi,u(\varphi)),\iota_{\ast}h(\varphi)\right)\right]\mathrm{d}\varphi=0\label{eq:quad_mainproperty}
\end{gather}
where $\mathrm{d}\varphi:=\mathrm{d}\varphi_{1}\wedge\ldots\wedge\mathrm{d}\varphi_{k}$
is differential volume form on torus $\mathbb{T}^{k}$.\end{prop}
\begin{proof}
By definition of classical solution,
\begin{gather}
\nabla_{\dot{x}(t)}\dot{x}(t)=\nabla W(t\omega,x(t))\quad\forall t\in\mathbb{R}.\label{eq:quat_LagrEquat}
\end{gather}
If for any point $x\in\mathcal{M}$, we define the orthogonal projection
operator
\begin{gather*}
P_{x}:T_{\iota(x)}\mathbb{E}^{n}\mapsto\iota_{\ast}T_{x}\mathcal{M},
\end{gather*}
then by the well-known property of the Levi--Civita connection \cite[Sect. 3.5]{GKM71}
for any smooth field of tangent vectors $\xi(t)$ along $x(t)$ we
have
\begin{gather}
\iota_{\ast}\nabla_{\dot{x}(t)}\xi(t)=P_{x(t)}\frac{\mathrm{d}}{\mathrm{d}t}\iota_{\ast}\xi(t).\label{eq:quat_covder}
\end{gather}
Now taking into account rational independence of frequency vector components
the equality \eqref{eq:quat_LagrEquat} implies the chain of equivalent
interrelations
\begin{gather}
P_{x(t)}\frac{\mathrm{d}}{\mathrm{d}t}\iota_{\ast}\dot{x}(t)=\iota_{\ast}\nabla W(t\omega,x(t))\quad\forall t\in\mathbb{R}\quad\Leftrightarrow\nonumber \\
P_{u(\varphi)}D_{\omega}^{2}\iota\circ u(\varphi)=\iota_{\ast}\nabla W(\varphi,u(\varphi))\quad\forall\varphi\in\mathbb{T}^{k}\quad\Leftrightarrow\nonumber \\
\intop_{\mathbb{T}^{k}}\left(D_{\omega}^{2}\iota\circ u(\varphi),P_{u(\varphi)}v(\varphi)\right)\mathrm{d}\varphi=\intop_{\mathbb{T}^{k}}\left(\iota_{\ast}\nabla W(\varphi,u(\varphi)),P_{u(\varphi)}v(\varphi)\right)\mathrm{d}\varphi\nonumber \\
\forall v(\cdot)\in\mathrm{C}\!\left(\mathbb{T}^{k};\mathbb{E}^{n}\right):\; D_{\omega}v(\cdot)\in\mathrm{C}\!\left(\mathbb{T}^{k};\mathbb{E}^{n}\right)\quad\Leftrightarrow\nonumber \\
\intop_{\mathbb{T}^{k}}\left(D_{\omega}^{2}\iota\circ u(\varphi),\iota_{\ast}h(\varphi)\right)\mathrm{d}\varphi=\intop_{\mathbb{T}^{k}}\left(\iota_{\ast}\nabla W(\varphi,u(\varphi)),\iota_{\ast}h(\varphi)\right)\mathrm{d}\varphi\label{eq:quat_D2om}\\
\forall h(\cdot)\in\mathrm{C}\!\left(\mathbb{T}^{k};T\mathcal{M}\right):\; D_{\omega}\iota_{\ast}h(\cdot)\in\mathrm{C}\!\left(\mathbb{T}^{k};\mathbb{E}^{n}\right),\; h(\varphi)\in T_{u(\varphi)}\mathcal{M}\;\forall\varphi\in\mathbb{T}^{k}.\nonumber
\end{gather}
Integrating by parts of the left-hand side of equality \eqref{eq:quat_D2om}
we get the required result.
\end{proof}
Now, basing on Proposition~\ref{prop:quat_0_1}, introduce the notion
of weak quasiperiodic solution to the system with Lagrangian density~\eqref{eq:quat_Lagr}
on Riemannian manifold $\mathcal{M}$. Before formulating the corresponding
definition let us recall definitions of a number of functional spaces
and corresponding objects.

Denote by $\mathrm{H}(\mathbb{T}^{k};\mathbb{E}^{n})$ the space of
$\mathbb{E}^{n}$-valued functions on torus $\mathbb{T}^{k}$ Lebesgue
integrable by the square of standard Euclidean norm $\left\Vert \cdot\right\Vert =\sqrt{(\cdot,\cdot)}$
in $\mathbb{E}^{n}$ (note that in view of the isometricity of embedding
$\iota$ the norms of vectors $\xi\in T_{x}\mathcal{M}$ and $\iota_{\ast}\xi\in T_{\iota(x)}\mathbb{E}^{n}\sim\mathbb{E}^{n}$
in corresponding spaces are equal, thus we use the same symbols for
these norms). For elements of $\mathrm{H}(\mathbb{T}^{k};\mathbb{E}^{n})$,
one can define a dot product $\left(\cdot,\cdot\right)_{0}:=(2\pi)^{-k}\intop_{\mathbb{T}^{k}}(\cdot,\cdot)\mathrm{d}\varphi$
as well as seminorm $\left\Vert \cdot\right\Vert _{0}=\sqrt{(\cdot,\cdot)_{0}^{2}}$.
By $\mathrm{H}_{\omega}^{1}(\mathbb{T}^{k};\mathbb{E}^{n})$ denote
the space of functions $f(\cdot)\in\mathrm{H}(\mathbb{T}^{k};\mathbb{E}^{n})$
each of which has weak (generalized) Sobolev derivative $D_{\omega}f(\cdot)\in\mathrm{H}(\mathbb{T}^{k};\mathbb{E}^{n})$
in direction of $\omega$, i.e. $\left(f(\cdot),D_{\omega}g(\cdot)\right)_{0}=-\left(D_{\omega}f(\cdot),g(\cdot)\right)_{0}$ for any
continuously differentiable function $g(\cdot):\mathbb{T}^{k}\mapsto\mathbb{E}^{n}$.
In the space $\mathrm{H}_{\omega}^{1}(\mathbb{T}^{k};\mathbb{E}^{n})$
one can naturally define a dot product $(\cdot,\cdot)_{1}:=\left(\cdot,\cdot\right)_{0}+\left(D_{\omega}\cdot,D_{\omega}\cdot\right)_{0}$
and a seminorm $\left\Vert \cdot\right\Vert _{1}$ generated by it.
After we identify elements equal a.e., both spaces $\left(\mathrm{H}(\mathbb{T}^{k};\mathbb{E}^{n}),\left(\cdot,\cdot\right)_{0}\right)$,
$\left(\mathrm{H}_{\omega}^{1}(\mathbb{T}^{k};\mathbb{E}^{n}),\left(\cdot,\cdot\right)_{1}\right)$
become Hilbertian ones (see, e.g., \cite{Sam91,BloPen01}).

Next, for arbitrary bounded set $\mathcal{A}\subset\mathcal{M}$,
denote by $\mathcal{S}_{\mathcal{A}}$ the space of smooth functions
$u(\cdot):\mathbb{T}^{k}\mapsto\mathcal{A}$.
\begin{defn}
A function $u(\cdot):\mathbb{T}^{k}\mapsto\mathcal{M}$ is said to
be of class $\mathcal{H}_{\mathcal{A}}^{1}$, if $\iota\circ u(\cdot)$
is a strong limit in $\mathrm{H}_{\omega}^{1}(\mathbb{T}^{k};\mathbb{E}^{n})$
of a sequence $\left\{ \iota\circ u_{j}(\cdot)\right\} $, where $u_{j}(\cdot)\in\mathcal{S}_{\mathcal{A}}$,
$j=1,2,\ldots$.
\end{defn}
Note that for compact set $\bar{\mathcal{A}}$$\subset\mathcal{M}$
there exist positive constants $c$ and $C$ such that
\begin{gather*}
c\left\Vert \iota(x_{1})-\iota(x_{2})\right\Vert \le\rho(x_{1},x_{2})\le C\left\Vert \iota(x_{1})-\iota(x_{2})\right\Vert \quad\forall x_{1},x_{2}\in\bar{\mathcal{A}}.
\end{gather*}
Without loss of generality one can regard functions of class $\mathcal{H}_{\mathcal{A}}^{1}$
as taking values in $\bar{\mathcal{A}}$.
\begin{defn}
\label{def:quad_v_f} A mapping $h(\cdot):\mathbb{T}^{k}\mapsto T\mathcal{M}$
is called a vector field along a mapping $u(\cdot)\in\mathcal{H}_{\mathcal{A}}^{1}$
determined by a sequence $\left\{ u_{j}(\cdot)\in\mathcal{S}_{\mathcal{A}}\right\} $
if $\iota_{\ast}h(\cdot)$ is a strong limit in $\mathrm{H}_{\omega}^{1}\left(\mathbb{T};\mathbb{E}^{n}\right)$
of a sequence $\left\{ \iota_{\ast}h_{j}(\cdot)\right\} $, where
$\left\{ h_{j}(\cdot):\mathbb{T}^{k}\mapsto T\mathcal{M}\right\} $
is a sequence of smooth mappings such that
\[
h_{j}(\varphi)\in T_{u_{j}(\varphi)}\mathcal{M}\quad\forall j\in\mathbb{N},\quad\sup_{j\in\mathbb{N},\varphi\in\mathbb{T}^{k}}\left\Vert h_{j}(\varphi)\right\Vert <\infty.
\]

\end{defn}
By the well-known Riesz–Fischer theorem, a formal sum $\sum_{\mathbf{n}\in\mathbb{Z}^{k}}f_{\mathbf{n}}\mathrm{e}^{\mathrm{i}(\mathbf{n},\varphi)}$
is the Fourier series of a function $f(\cdot)\in\mathrm{H}(\mathbb{T}^{k};\mathbb{E}^{n})$
iff the series $\sum_{\mathbf{n}\in\mathbb{Z}^{k}}\left\Vert f_{\mathbf{n}}\right\Vert ^{2}$
is convergent. On the other hand, by the Riesz-Fischer-Besicovitch
theorem \cite[p. 110]{Bez55}, to the formal sum $\sum_{\mathbf{n}\in\mathbb{Z}^{k}}f_{\mathbf{n}}\mathrm{e}^{\mathrm{i}(\mathbf{n},\omega)t}$,
one can assign (nonuniquely) a function $t\mapsto\mathfrak{f}(t)$
of class $L_{\mathrm{loc}}^{2}(\mathbb{R};\mathbb{E}^{n})$ for which
the above sum is its Fourier series and which represents an element
of the space of Besicovitch quasiperiodic functions $\mathrm{B}^{2}\!\left(\mathbb{R};\mathbb{E}^{n}\right)$.
In view of the form of both Fourier series we generally use the notation
$\mathfrak{f}(t)=f(t\omega)$ for the quasiperiodic function $\mathfrak{f}(\cdot)$,
associated with a function $f(\cdot)$ on torus. However, one should
keep in mind the following factors. Firstly, when proving the Riesz-Fischer-Besicovitch
theorem a Besicovitch function with prescribed Fourier series is constructed
irrespectively to restriction of concrete associated function to line$\left\{ \varphi=t\omega\right\} _{t\in\mathbb{R}}$,
and, besides, two functions from $L_{\mathrm{loc}}^{2}(\mathbb{R};\mathbb{E}^{n})$
representing an element of the space $\mathrm{B}^{2}\!\left(\mathbb{R};\mathbb{E}^{n}\right)$
(i.e. possessing common Fourier series) may differ by everywhere nonvanishing
function. Secondly, if we take a function $f(\cdot)\in\mathrm{H}(\mathbb{T}^{k};\mathbb{E}^{n})$,
then we can only assert that the function $t\mapsto f(\varphi+t\omega)$
belongs to $\mathrm{B}^{2}\!\left(\mathbb{R};\mathbb{E}^{n}\right)$
for almost every point $\varphi\in\mathbb{T}^{k}$, but a priori it's
unknown whether the same is true for $\varphi=0$. To give informal
sense to the equality $\mathfrak{f}(t)=f(t\omega)$, conversely, having
a function $\mathfrak{f}(\cdot)\in\mathrm{B}^{2}\!\left(\mathbb{R};\mathbb{E}^{n}\right)$,
we think of associated function $f(\cdot)\in\mathrm{H}(\mathbb{T}^{k};\mathbb{E}^{n})$
as being redefined on the line $\left\{ \varphi=t\omega\right\} _{t\in\mathbb{R}}$
(which has zero Lebesgue measure) in such a way that the above equality
is actually fulfilled \cite[Sect. 1.5]{Sam91}.

Now let $f(\cdot)=\iota\circ u(\cdot)$, where $u(\cdot)\in\mathcal{H}_{\mathcal{A}}^{1}$.
To this function, one can assign an associated $\mathbb{E}^{n}$-valued
Besicovitch quasiperiodic function. But in view of the above mentioned
circumstances the following question needs to be clarified: does there
exist a Besicovitch quasiperiodic function associated with $\iota\circ u(\cdot)$
and taking values in $\iota(\bar{\mathcal{A}})$? I.e., does there
exist a Besicovitch quasiperiodic function associated with $u(\cdot)\in\mathcal{H}_{\mathcal{A}}^{1}$
and taking values in $\bar{\mathcal{A}}$? Basing on completeness of Marcinkiewicz spaces (see, \cite{Mar39,BohFol45,Dan06}) one can show that the answer
to the question is affirmative.

Now in view of Proposition~\ref{prop:quat_0_1} we can introduce
the following definition.
\begin{defn}
\label{def:weakQPS}A quasiperiodic Besicovitch function $t\mapsto u(t\omega)$
associated with function $u(\cdot)\in\mathcal{H}_{\mathcal{A}}^{1}$,
where $\mathcal{A}\subseteq\mathcal{M}$ is a bounded set, is called
a weak solution of Lagrangian system on $\mathcal{M}$ with Lagrangian
density \eqref{eq:quat_Lagr} if
\begin{gather}
\left(D_{\omega}\iota\circ u(\cdot),D_{\omega}\iota_{\ast}h(\cdot)\right)_{0}+\left(\iota_{\ast}\nabla W(\cdot,u(\cdot)),\iota_{\ast}h(\cdot)\right)_{0}=0\label{eq:quat_1varJ0}
\end{gather}
 for any vector field $h(\cdot)$ along $u(\cdot)$.

The function $u(\cdot)$ from this definition can be interpreted as
an extremal of functional
\begin{gather*}
J[u(\cdot)]:=\intop_{\mathbb{T}^{k}}\left[\frac{1}{2}\left\Vert D_{\omega}\iota\circ u(\varphi)\right\Vert ^{2}+W(\varphi,u(\varphi))\right]\mathrm{d}\varphi
\end{gather*}
on set $\mathcal{H}_{\mathcal{A}}^{1}$, and the corresponding weak
solution $t\mapsto u(t\omega)$ as an extremal of functional $\mathcal{L}$
\eqref{eq:quat_Lagr_lim} on set of Besicovitch quasiperiodic functions
associated with functions of class $\mathcal{H}_{\mathcal{A}}^{1}$.
\end{defn}
Now let us prove an existence theorem for weak quasiperiodic solution.
\begin{thm}
\label{thm:quat_1} Let there exist a bounded smooth function $V(\cdot):\mathcal{M}\mapsto\mathbb{R}$
satisfying conditions\textbf{ C1},\textbf{C2}, as well as the inequalities

\begin{gather*}
\lambda_{W}(\varphi,x)+\frac{1}{2}\left\langle \nabla W(\varphi,x),\nabla V(x)\right\rangle >0\quad\forall(\varphi,x)\in\mathbb{T}^{k}\times\bar{\Omega},\\
\left\langle \nabla W(\varphi,x),\nabla V(x)\right\rangle >0\quad\forall(\varphi,x)\in\mathbb{T}^{k}\times\partial\Omega,
\end{gather*}
where $\lambda_{W}(\varphi,x):=\min\left\{ \left\langle H_{W}(\varphi,x)\xi,\xi\right\rangle :\xi\in T_{x}\mathcal{M},\left\Vert \xi\right\Vert =1\right\} $,
$H_{W}(\varphi,x)$ is the Hessian of function $W(\varphi,\cdot):\mathcal{M}\mapsto\mathbb{R}$
at point $x$ for $\varphi\in\mathbb{T}^{k}$. Then the system with
Lagrangian density \eqref{eq:quat_Lagr} has a weak quasiperiodic
solution associated with function of class $\mathcal{H}_{\Omega}^{1}$.\end{thm}
\begin{proof}
At first, let us show that any two points of the set $\Omega:=V^{-1}((-\infty,v))$
can be connected in $\Omega$ by a geodesic segment of conformally
equivalent metric $\left\langle \cdot,\cdot\right\rangle _{V}:=\mathrm{e}^{V}\left\langle \cdot,\cdot\right\rangle $.
Denote by $\nabla^{V}$ the Levi--Civita connection of Riemannian
manifold $\left(\mathcal{M},\left\langle \cdot,\cdot\right\rangle _{V}\right)$,
and by $\exp_{x}^{V}(\cdot)$ the exponential mapping at $x$ associated
with $\nabla^{V}$. Let $x\in\Omega$ and thus $v_{0}:=V(x)<v$. Consider
the open set
\begin{gather*}
\mathcal{Z}_{x}=\left\{ \xi\in T_{x}\mathcal{M}:\exp_{x}^{V}(s\xi)\in\Omega\;\forall s\in[0,1]\right\} .
\end{gather*}
In \cite[p. 10]{ParRus12}, relying on well-known formula of sectional
curvature for conformally equivalent metric \cite[Sect. 3.6]{GKM71},
it is observed that the inequality~\eqref{eq:conv_mainineq} yields
nonpositiveness of sectional curvature for connection $\nabla^{V}$
at any point of $\Omega$ and along any two-dimensional direction.
Then by the Morse–Schoenberg theorem \cite[Sect. 6.2]{GKM71}, for
any $\xi\in\mathcal{Z}_{x}$, the geodesic segment $\exp_{x}^{V}(s\xi)$,
$s\in[0,1]$, does not contain conjugate points. For this reason the
mapping $\exp_{x}^{V}(\cdot)$ restricted to the set $\mathcal{Z}_{x}$
is a local diffeomorphism and this implies that $\Xi:=\exp_{x}^{V}(\mathcal{Z}_{x})$
is an open subset of the set $\Omega$.

Let us show that at the same time $\Xi$ is a closed subset of $\Omega$.
Assume $\left\{ \xi_{k}\in\mathcal{Z}_{x}\right\} $ to be such a
sequence that the subsequence $\left\{ x_{k}=\exp_{x}^{V}(\xi_{k})\in\Xi\right\} $
converges to a point $x_{\ast}\in\Omega$. We must show the existence
of a vector $\xi_{\ast}\in\mathcal{Z}_{x}$ such that $\exp_{x}^{V}(\xi_{\ast})=x_{\ast}$.
The set $\bar{\mathcal{Z}}_{x}$ is compact, so we may consider that
$\xi_{k}$ converges to a point $\xi_{\ast}\in\bar{\mathcal{Z}}_{x}$.
It remains to show that $\xi_{\ast}\not\in\partial\mathcal{Z}_{x}$.
Suppose, conversely, that $\xi_{\ast}\in\partial\mathcal{Z}_{x}$.
Since for sufficiently small $\delta>0$ and for all $t\in[0,\delta]$
we have $\exp_{x}^{V}(t\xi_{\ast})\in\Omega$, then $t\xi_{\ast}\in\mathcal{Z}_{x}$
for all $t\in[0,\delta]$. From definition of $\mathcal{Z}_{x}$ it
follows that there exists $s_{\ast}\in(\delta,1)$ such that $\exp_{x}^{V}(s_{\ast}\xi_{\ast})\in\partial\Omega$,
but $\exp_{x}^{V}(s\xi_{\ast})\in\Omega$ when $s\in[0,s_{\ast})$.
Here we take into account that the equality $s_{\ast}=1$ is impossible
since $x_{\ast}\in\Omega$ (if the equality $s_{\ast}=1$ were valid,
the vector $\xi_{\ast}$ would belong to $\mathcal{Z}_{x}$ in contrary
to our assumption).

Consider now the sequence $s_{\ast}\xi_{k}$. It converges to $s_{\ast}\xi_{\ast}$.
Hence, $V\circ\exp_{x}^{V}(s_{\ast}\xi_{k})\to v$ and $V\circ\exp_{x}^{V}(s_{\ast}\xi_{k})<v$.
Put $\hat{V}(\xi):=\exp\circ V\circ\exp_{x}^{V}(\xi)$. Then for all
sufficiently large $k$ we have
\begin{gather}
\frac{1}{2}\left(\mathrm{e}^{v}+\mathrm{e}^{v_{0}}\right)<\hat{V}(s_{\ast}\xi_{k})<\mathrm{e}^{v}.\label{eq:conv_Vlev}
\end{gather}
 It is shown in \cite[p. 9]{ParRus12} that the inclusion $\bar{\Omega}\subset\mathcal{D}$
yields the existence of $\sigma>0$ such that
\begin{gather*}
\frac{\mathrm{d}^{2}}{\mathrm{d}s^{2}}\hat{V}(s\xi_{k})\ge\sigma\quad\forall s\in[0,1],\;\forall k\in\mathbb{N}.
\end{gather*}
Then each derivative $\frac{\mathrm{d}}{\mathrm{d}s}\hat{V}(s\xi_{k})$
is monotone increasing on segment $[0,1]$ and in addition
\begin{equation}
\frac{\mathrm{d}}{\mathrm{d}s}\hat{V}(s\xi_{k})\ge\frac{1}{2s_{\ast}}\left(\mathrm{e}^{v}-\mathrm{e}^{v_{0}}\right)\quad\forall s\in[s_{\ast},1]\label{eq:conv_dVds_gev-v0}
\end{equation}
 for all $k$ starting from sufficiently large number $k_{\ast}$.
In fact, in opposite case, there would exists sufficiently large $k$,
for which the inequalities ~\eqref{eq:conv_Vlev} hold true together
with
\[
\frac{\mathrm{d}}{\mathrm{d}s}\hat{V}(s\xi_{k})\le\frac{\mathrm{d}}{\mathrm{d}s}\Bigl|_{s=s_{\ast}}\hat{V}(s\xi_{k})<\frac{1}{2s_{\ast}}\left(\mathrm{e}^{v}-\mathrm{e}^{v_{0}}\right)\quad\forall s\in[0,s_{\ast}].
\]
But then
\begin{gather*}
\hat{V}(s_{\ast}\xi_{k})<\mathrm{e}^{v_{0}}+\frac{1}{2s_{\ast}}\left(\mathrm{e}^{v}-\mathrm{e}^{v_{0}}\right)s_{\ast}=\frac{1}{2}\left(\mathrm{e}^{v}+\mathrm{e}^{v_{0}}\right)
\end{gather*}
in contrary to inequality~\eqref{eq:conv_Vlev}. Now from inequality~\eqref{eq:conv_dVds_gev-v0}
we get
\begin{gather*}
\exp\circ V(x_{k})=\hat{V}(\xi_{k})\ge\hat{V}(s_{\ast}\xi_{k})+\frac{1}{2s_{\ast}}\left(\mathrm{e}^{v}-\mathrm{e}^{v_{0}}\right)(1-s_{\ast})\quad\forall k\ge k_{\ast}.
\end{gather*}
 Letting $k$ to infinity we arrive at
\begin{gather*}
\exp\circ V(x_{\ast})\ge\mathrm{e}^{v}+\frac{1}{2s_{\ast}}\left(\mathrm{e}^{v}-\mathrm{e}^{v_{0}}\right)(1-s_{\ast})>\mathrm{e}^{v}\quad\Rightarrow\quad V(x_{\ast})>v,
\end{gather*}
and this contradicts with assumption that $x_{\ast}\in\Omega$, i.e,
that $V(x_{\ast})<v$.

Hence, $\Xi$ is open--closed (in $\Omega$) subset of open set $\Omega$.
Since, by assumption, the set $\Omega$ is connected, then, as is
commonly known, $\Xi=\Omega$.

Thus for any pair of points $x,y\in\Omega$ there exists a tangent
vector $\zeta(x,y)\in T_{x}\mathcal{M}$ such that $\mathrm{exp}_{x}^{V}(\zeta(x,y))=y$.
Furthermore, the geodesic segment $\bigcup_{s\in[0,1]}\mathrm{exp}_{x}^{V}(s\zeta(x,y))$
connects points $x$ and $y$ and completely belongs to $\Omega$.
Moreover, it is shown in \cite[Proposition 3.8]{ParRus12} that the
mapping
\begin{gather*}
\mathrm{exp}_{x}^{V}(\cdot):\mathcal{Z}_{x}\mapsto\Omega
\end{gather*}
 is a diffeomorphism, and hence, the points $x,y$ uniquely determine
the vector $\zeta(x,y)$, as well as the geodesic segment of metric
$\left\langle \cdot,\cdot\right\rangle _{V}$ connecting the points
and lying in $\Omega$.

The further proving relies on propositions 3.9–3.11 from \cite{ParRus12}
and completely duplicates reasonings, contained in proof of Theorem
4.1 and in Addendum to the paper cited. On the basis of these reasonings
we can establish the following: 1) there exists a smooth mapping $\chi(\cdot,\cdot,\cdot):[0,1]\times\Omega\times\Omega\mapsto\Omega$
(connecting mapping) such that for fixed $x,y\in\Omega$ the function
$\chi(\cdot,x,y):[0,1]\mapsto\Omega$ is a solution to equation
\begin{gather}
\nabla_{x^{\prime}}x^{\prime}=\frac{\left\Vert x^{\prime}\right\Vert ^{2}}{2}\nabla V(x)\quad\left(x^{\prime}:=\frac{\mathrm{d}x}{\mathrm{d}s}\right)\text{,}\label{eq:quat_prop_chi}
\end{gather}
satisfies $\chi(0,x,y)=x$, $\chi(1,x,y)=y$, and one can point out
a positive number $\varkappa>0$ such that for any pair of smooth
functions $x_{i}(\cdot):\mathbb{R}\mapsto\Omega$, $i=1,2$, and for
any $\varphi\in\mathbb{T}^{k}$ the convexity condition of Lagrangian
density \eqref{eq:quat_Lagr_phi} is fulfilled
\begin{gather}
\begin{split} & \frac{\partial^{2}}{\partial s^{2}}L\left(\varphi+t\omega,\chi\left(s,x_{1}(t),x_{2}(t)\right),\frac{\partial}{\partial t}\chi\left(s,x_{1}(t),x_{2}(t)\right)\right)\ge\\
 & \phantom{{\}\frac{\partial}{\partial^{2}}L_{\varphi}(t\omega+\varphi)}}\ge\varkappa\left[\left\Vert \nabla_{\xi}\eta\right\Vert ^{2}+\left\Vert \xi\right\Vert ^{2}\left(\left\Vert \eta\right\Vert ^{2}+1\right)\right],
\end{split}
\label{eq:quat_convexL}
\end{gather}
where
\begin{gather*}
\eta:=\eta(s,t):=\frac{\partial}{\partial t}\chi\left(s,x_{1}(t),x_{2}(t)\right),\quad\xi:=\xi(s,t):=\frac{\partial}{\partial s}\chi\left(s,x_{1}(t),x_{2}(t)\right)
\end{gather*}
are vector fields along the mapping $(s,t)\mapsto\chi\left(s,x_{1}(t),x_{2}(t)\right)$;
2) as a consequence, for arbitrary $u_{1}(\cdot),u_{2}(\cdot)\in\mathcal{S}_{\Omega}$
the function
\begin{gather*}
s\mapsto J[\chi(s,u_{1}(\cdot),u_{2}(\cdot))]
\end{gather*}
is strictly convex downward on segment $[0,1]$; 3) a minimizing sequence
$\left\{ u_{j}(\cdot)\in\mathcal{S}_{\Omega}\right\} $
for $J\bigl|_{\mathcal{S}_{\Omega}}$ converges to a function $u(\cdot)\in\mathcal{H}_{\Omega}^{1}$
in the sense that
\begin{gather*}
\lim_{j\to\infty}\left\Vert \iota\circ u_{j}(\cdot)-\iota\circ u(\cdot)\right\Vert _{1}=0,
\end{gather*}
the equality $\inf\left\{ J[\mathcal{S}_{\Omega+\delta}]\right\} =J[u(\cdot)]$
being valid for sufficiently small $\delta>0$; 3) the function $u(\cdot)$
satisfies the equality $J^{\prime}[u_{\ast}(\cdot)]h(\cdot)=0$ (equivalent
to~\eqref{eq:quat_1varJ0}) for arbitrary vector field $h(\cdot)$
along $u(\cdot)$. This means that $t\mapsto u(t\omega)$ is a weak
quasiperiodic solution to system with Lagrangian density~\eqref{eq:quat_Lagr}. \end{proof}
\begin{rem}
\label{rem:quat_value_u}Since the inequalities satisfying by functions
$V(\cdot)$ and $W(\cdot,\cdot)$ in accordance with conditions of
Theorem~\ref{thm:quat_1} are strict, the result obtained is still
correct for a domain $\Omega^{\prime}=V^{-1}(-\infty,v^{\prime})\subset\Omega$
with arbitrary $v^{\prime}<v$ sufficiently close to $v$. And then
$u(\cdot)\in\mathcal{H}_{\Omega^{\prime}}^{1}$.
\end{rem}
\textbf{3. An existence theorem for classical quasiperiodic solution.}

The main result of this paper is the following theorem
\begin{thm}
\label{thm:quat_class_QPS}Let there hold the condition of Theorem~\ref{thm:quat_1}
and a function $u(\cdot)\in\mathcal{H}_{\Omega}^{1}$ defines a weak
quasiperiodic solution to system with Lagrangian density~\eqref{eq:quat_Lagr}
Then $u(\cdot)\in\mathrm{C}\!\left(\mathbb{T}^{k};{\Omega}\right)$
and the function $x(t):=u(t\omega)$ is a classical uniformly quasiperiodic
solution to this system.
\end{thm}
The proof of the theorem relies on propositions~\ref{prop:quat_classical}--\ref{prop:quat_uniqBS}
given below.

At first we exploit a technique proposed in \cite{BerZha96} to prove
the following proposition.
\begin{prop}
\label{prop:quat_classical} Let there hold the conditions of Theorem~\ref{thm:quat_1}
and the function $u(\cdot)\in\mathcal{H}_{\Omega}^{1}$ defines a
weak quasiperiodic solution of the system with Lagrangian density~\eqref{eq:quat_Lagr}.
Then the function $t\mapsto u(\varphi+t\omega)$ is a classical Besicovitch
quasiperiodic solution to system with Lagrangian density
\begin{gather}
L(\varphi+t\omega,x,\dot{x}):=\mathrm{K}(\dot{x})+W(\varphi+t\omega,x)\label{eq:quat_Lagr_phi}
\end{gather}
 for almost all $\varphi\in\mathbb{T}^{k}$. This solution takes values
in a compact set $\mathcal{K}\subset\Omega$.\end{prop}
\begin{proof}
Relying on reasonings from the proof of Theorem~\ref{thm:quat_1}
we can conclude that the set $\Omega$ is diffeomorphic to a domain
$\mathcal{U}$ of Euclidean space $\mathbb{E}^{m}=\left(\mathbb{R}^{m},\left(\cdot,\cdot\right)\right)$.
We may think of this domain as a map of the set $\Omega$. Thus, taking
into account Proposition~\ref{rem:quat_value_u}, we shall consider
that the function $u(\cdot)$, existence of which is established Theorem~\ref{thm:quat_1},
belongs to class $\mathrm{H}_{\omega}^{1}\left(\mathbb{T}^{k};\mathbb{E}^{m}\right)$
and takes values in a compact subset $\mathcal{C}\subset\mathcal{U}$. We shall also consider that the embedding $\iota(\cdot)$ acts from $\mathcal{U}$ into $\mathbb{E}^{n}$.
The metric \textbf{$\left\langle \cdot,\cdot\right\rangle $ }induce
in $\mathcal{U}$ a tensor field $g(x)$ and the corresponding metric
$\left(g(x)\cdot,\cdot\right)=\left\langle \cdot,\cdot\right\rangle $.
In view of~\eqref{eq:quat_covder}, for any smooth vector field of
tangent vectors $\xi(t)$ along smooth curve $x(t)$ there hold equalities
\begin{gather*}
\left(\frac{\mathrm{d}}{\mathrm{d}t}\iota x(t),\frac{\mathrm{d}}{\mathrm{d}t}\iota_{\ast}\xi(t)\right)=\left(\iota_{\ast}\dot{x}(t),\frac{\mathrm{d}}{\mathrm{d}t}\iota_{\ast}\xi(t)\right)=\left(\iota_{\ast}\dot{x}(t),P_{x(t)}\frac{\mathrm{d}}{\mathrm{d}t}\iota_{\ast}\xi(t)\right)=\left(\iota_{\ast}\dot{x}(t),\iota_{\ast}\nabla_{\dot{x}(t)}\xi(t)\right),
\end{gather*}
 and in  coordinates of domain $\mathcal{U}$ we have
\begin{gather*}
\nabla_{\dot{x}(t)}\xi(t)=\dot{\xi}(t)+\Gamma_{x(t)}\left(\dot{x}(t),\xi(t)\right),
\end{gather*}
where a bilinear mapping $\Gamma_{x}:\mathbb{E}^{m}\times\mathbb{E}^{m}\mapsto\mathbb{E}^{m}$
smoothly depends on $x\in\mathcal{U}$ and can be expressed via
the Christoffel symbols. Therefore a sequence $\left\{ u_{j}(\cdot)\right\} $
which defines $u(\cdot)$ and a sequence $\left\{ h_{j}(\cdot)\right\} $
which according to Definition~\ref{def:quad_v_f} defines a vector
field $h(\cdot)$ along $u(\cdot)$ satisfy the equalities
\begin{gather*}
\left(D_{\omega}\iota\circ u_{j}(\varphi),D_{\omega}\iota_{\ast}h_{j}(\varphi)\right)=\left(\iota_{\ast}D_{\omega}u_{j}(\varphi),\iota_{\ast}\nabla_{D_{\omega}u_{j}(\varphi)}h_{j}(\varphi)\right)=\\
=\left(g\left(u_{j}(\varphi)\right)D_{\omega}u_{j}(\varphi),\nabla_{D_{\omega}u_{j}(\varphi)}h_{j}(\varphi)\right)=\\
=\left(g\left(u_{j}(\varphi)\right)D_{\omega}u_{j}(\varphi),D_{\omega}h_{j}(\varphi)+\Gamma_{u_{j}(\varphi)}\left(D_{\omega}u_{j}(\varphi),h_{j}(\varphi)\right)\right).
\end{gather*}
 Consider a special case where the sequence $\left\{ h_{j}(\cdot)\right\} $
is defined by a unique smooth mapping $h(\cdot):\mathbb{T}^{k}\mapsto\mathbb{E}^{m}$,
i.e., the vector $h_{j}(\varphi)$ has an origin at point $u_{j}(\varphi)$
and an end at point $u_{j}(\varphi)+h(\varphi)$. Taking into account
that $\left\{ u_{j}(\cdot)\right\} $ is uniformly bounded and converges
in $\mathrm{H}_{\omega}^{1}\!\left(\mathbb{T}^{k};\mathbb{E}^{m}\right)$
and also that $\nabla W(\varphi,x)=g^{-1}(x)W_{x}^{\prime}(\varphi,x)$,
one can present the property~\eqref{eq:quad_mainproperty} of function
$u(\cdot)$ in the form
\begin{gather}
\intop_{\mathbb{T}^{k}}\left[\left(g(u(\varphi))D_{\omega}u(\varphi),D_{\omega}h(\varphi)+\Gamma_{u(\varphi)}\left(D_{\omega}u(\varphi),h(\varphi)\right)\right)+\left(W_{x}^{\prime}(\varphi,u(\varphi)),h(\varphi)\right)\right]\mathrm{d}\varphi=0.\label{eq:quat_1stvareq0}
\end{gather}

Just as in \cite{BerZha96}, in the integral from the left-hand side,
we perform a change of variables $\varphi\to(\tau,y)$ by formula
\begin{gather*}
\varphi=Q(\tau,y):=\sum_{i=1}^{k-1}y_{i}\varepsilon_{i}+\tau\varepsilon_{k},
\end{gather*}
 where  $\left\{ \varepsilon_{i}\right\} _{i=1}^{k}$
is an orthonormal basis in $\mathbb{E}^{k}$ with $\varepsilon_{k}:=\omega/\left\Vert \omega\right\Vert $, and $y=(y_{1},\ldots,y_{k-1})$,
Setting
\begin{gather*}
v(\tau,y):=u(Q(\tau,y)),\quad w(\tau,y):=h(Q(\tau,y)),\quad I(y)\text{:}=\left\{ \tau\in\mathbb{R}:Q(\tau,y)\in K\right\} ,
\end{gather*}
and applying the Fubini theorem to the left-hand side of the equality
\begin{gather}
\intop_{K}\left[\left(g(u(\varphi))D_{\omega}u(\varphi),D_{\omega}h(\varphi)+\Gamma_{u(\varphi)}\left(D_{\omega}u(\varphi),h(\varphi)\right)\right)+\left(W_{x}^{\prime}(\varphi,u(\varphi)),h(\varphi)\right)\right]\mathrm{d}\varphi=0,\label{eq:quat_1stvar_K}
\end{gather}
 we get
\begin{gather}
\intop_{Y}\intop_{I(y)}\left[\left(\left\Vert \omega\right\Vert ^{2}g(v(\tau,y))v_{\tau}^{\prime}(\tau,u),\dot{w}(\tau,y)+\Gamma_{v(\tau,y)}\left(v_{\tau}^{\prime}(\tau,y),w(\tau,y)\right)\right)\right]\mathrm{d}\tau\mathrm{d}y+\nonumber \\
+\intop_{Y}\intop_{I(y)}\left(W_{x}^{\prime}(Q(\tau,y),v(\tau,y)),w(\tau,y)\right)\mathrm{d}\tau\mathrm{d}y=0,\label{eq:quat_gvwG}
\end{gather}
 where $v_{\tau}^{\prime}(\tau,u)$ stands for Sobolev generalized
partial derivative by variable $\tau$, and $\dot{w}(\tau,y)=\partial w(\tau,y)/\partial\tau$
is classical partial derivative by this variable. Since $u(\varphi)\in\mathrm{H}_{\omega}^{1}\left(\mathbb{T}^{k};\mathbb{E}^{m}\right)$,
then
\begin{gather*}
\intop_{Y}\intop_{I(y)}\left\Vert v(\tau,y)\right\Vert ^{2}+\left\Vert v_{\tau}^{\prime}(\tau,y)\right\Vert ^{2}\mathrm{d}\tau\mathrm{d}y\text{<}\infty.
\end{gather*}
If we now denote by $Y$ the orthogonal projection of $k$-dimensional
cube $K:=[0,2\pi]^{k}$ onto the hyperplane with basis$\left\{ \varepsilon_{i}\right\} _{i=1}^{k-1}$,
then the Fubini theorem together with Theorem~6 \cite[p. 399]{Nik83}
implies the following assertions: there exists a set $Y^{\prime}\subset Y$
such that $\mathrm{mes}\, Y^{\prime}=\mathrm{mes}\, Y$ and for any
$y\in Y^{\prime}$ the function $v(\cdot,y):I(y)\mapsto\mathcal{C}$
is absolutely continuous on $I(y)$, has  classical derivative $\dot{v}(\tau,y)=v_{\tau}^{\prime}(\tau,y)$
a.e. on this segment, satisfies the condition
\begin{gather*}
\intop_{I(y)}\left\Vert v(\tau,y)\right\Vert ^{2}+\left\Vert v_{\tau}^{\prime}(\tau,y)\right\Vert ^{2}\mathrm{d}\tau<\infty\quad\forall y\in Y^{\prime},
\end{gather*}
and, as a consequence, belongs to Sobolev space $\mathrm{H}^{1}\!\left(I(y);\mathbb{E}^{m}\right)$.

By introducing a ``momentum''
\begin{gather*}
p(\tau,y)=g(v(\tau,y))\dot{v}(\tau,u),
\end{gather*}
we represent the equality~\eqref{eq:quat_gvwG} in the form
\begin{gather}
\intop_{Y}\intop_{I(y)}\left(p(\tau,y),\dot{w}(\tau,y)\right)\mathrm{d}\tau\mathrm{d}y=\nonumber \\
=-\intop_{Y}\intop_{I(y)}\left(G_{v(\tau,y)}(\dot{v}(\tau,y),\dot{v}(\tau,y))+\left\Vert \omega\right\Vert ^{-2}W_{x}^{\prime}(Q(\tau,y),v(\tau,y)),w(\tau,y)\right)\mathrm{d}\tau\mathrm{d}y,\label{eq:quat_gvwqG}
\end{gather}
where a family of bilinear mappings $G_{x}(\cdot,\cdot):\mathbb{E}^{m}\times\mathbb{E}^{m}\mapsto\mathbb{E}^{m}$
smoothly dependent on $x\in\mathcal{U}$ is defined by
\begin{gather*}
\left(g(x)a,\Gamma_{x}(b,c)\right)=\left(G_{x}(a,b),c\right)\quad\forall a,b,c\in\mathbb{E}^{m}.
\end{gather*}
 Since the equality~\eqref{eq:quat_gvwqG} is fulfilled for any function
$w(\tau,y)$ from space $\mathrm{C}_{0}^{\infty}\!\left(Q^{-1}\left(K\right);\mathbb{E}^{m}\right)$
formed by smooth functions having their supports in interior of the
set $Q^{-1}\left(K\right)=\bigcup_{y\in Y}I(y)$, then the above equality
implies that the function $p(\tau,y)$ has integrable in $Q^{-1}(K)$
Sobolev generalized derivative
\begin{gather}
p_{\tau}^{\prime}(\tau,y)=G_{v(\tau,y)}(\dot{v}(\tau,y),\dot{v}(\tau,y))+\left\Vert \omega\right\Vert ^{-2}W_{x}^{\prime}(Q(\tau,y),v(\tau,y)).\label{eq:quat_gvveqGW}
\end{gather}
By using Theorem 6 \cite[p. 399]{Nik83} again, we arrive at conclusion
that there exists a set $Y^{\prime\prime}\subset Y^{\prime}$ such
that $\mathrm{mes}\, Y^{\prime\prime}=\mathrm{mes}\, Y^{\prime}$
and for any $y\in Y^{\prime\prime}$ there holds the equality
\begin{gather}
p(\tau,y)=p(\tau_{0},y)+\intop_{\tau_{0}}^{\tau}\left[G_{v(s,y)}(\dot{v}(s,y),\dot{v}(s,y))+\left\Vert \omega\right\Vert ^{-2}W_{x}^{\prime}(Q(s,y),v(s,y))\right]\mathrm{d}s,\label{eq:quat_ptauy}
\end{gather}
on segment $I(y)$ where $\tau_{0}\in I(y)$ is a fixed point. Hence,
for any $y\in Y^{\prime\prime}$, the functions $\tau\mapsto p(\tau,y)$
and $\tau\mapsto\dot{v}(\tau,y)=p(\tau,y)/g(v(\tau,y))$ are absolutely
continuous on the segment $I(y)$, and each of these functions has
an ordinary partial derivative by $\tau$ a. e. on $I(y)$ . Then
the function $\tau\mapsto v(\tau,y)$ is continuously differentiable,
and from~\eqref{eq:quat_ptauy} it follows that the function $\tau\mapsto p(\tau,y)$
is continuously differentiable. Then it becomes clear that $\tau\mapsto\dot{v}(\tau,y)$
is continuously differentiable as well. Thus, for a.e. $y\in Y$ the
function $\tau\mapsto v(\tau,y)$ is twice continuously differentiable
and satisfies the equality \eqref{eq:quat_gvveqGW} everywhere on
$I(y)$. Having determined from this equality the function $\ddot{v}(\cdot,\cdot)$
of variables $\tau,y$, we see that it is integrable on $Q^{-1}(K)$.

Finally, let us show that for almost all $y\in Y$ the functions $\tau\mapsto v(\tau,y)$
and $\tau\mapsto p(\tau,y)$ have the properties described above not
only on $I(y)$, but also on any segment of real axis $\mathbb{R}$.
Represent the space $\mathbb{E}^{k}$ in the form of union of cubes:
$\mathbb{E}^{k}=\bigcup_{\mathbf{m}\in\mathbb{Z}^{k}}\left(K+2\pi\mathbf{m}\right)$.
Then $\mathbb{R}=\bigcup_{\mathbf{m}\in\mathbb{Z}}I_{\mathbf{m}}(y)$
where
\begin{gather*}
I_{0}(y):=I(y),\quad I_{\mathbf{m}}(y):=\left\{ \tau\in\mathbb{R}:Q(\tau,y)\in K+2\pi\mathbf{m}\right\} .
\end{gather*}
Among all $I_{\mathbf{m}}(y)$, there are sets of zero measure and,
in particular, empty sets. Put
\begin{gather*}
Y_{\mathbf{m}}:=\left\{ y\in Y:\mathrm{mes}\, I_{\mathbf{m}}(y)>0\right\} ,\quad\mathbb{M}:=\left\{ \mathbf{m}\in\mathbb{Z}^{k}:Y_{\mathbf{m}}\ne\varnothing\right\} .
\end{gather*}
It is clear that for a fixed $\mathbf{\mathbf{m}}$ the set $Y_{\mathbf{m}}$
consists only of those $y\in Y$ for which the line $L(y)$ defined
in $\mathbb{E}^{k}$ by equation $\varphi=Q(\tau,y),\;\tau\in\mathbb{R}$,
intersects the cube $K+2\pi\mathbf{m}$ by segment which is not reduced
to single-point set, and $\mathbb{M}$ contains only those integer
vectors $\mathbf{m}$, for which the corresponding cube $K+2\pi\mathbf{m}$
intersects at least with one line $L(y)$, where $y\in Y$, and the
intersection is a segment or a point. It is easily seen that rational
independence of components of frequency vector $\omega$ implies that
any nonempty set $Y_{\mathbf{m}}$ is open. Besides, $\mathrm{mes}(Y\setminus Y_{\mathbf{0}})=0$.

Having introduced the orthoprojectors
\begin{gather*}
\mathrm{pr}_{Y}(a):=\sum_{i=1}^{k-1}(a,\varepsilon_{i})\varepsilon_{i},\quad\mathrm{pr}_{T}(a):=(a,\varepsilon_{k})\varepsilon_{k}\quad\forall a\in\mathbb{E}^{k},
\end{gather*}
we get the following equalities %
\footnote{We identify a function on torus $\mathbb{T}^{k}=\mathbb{R}^{k}/2\pi\mathbb{Z}^{k}$
with its natural lift into the space $\mathbb{E}^{k}$ covering this
torus.%
}

\begin{gather*}
f(\tau,y):=F(Q(\tau,y))=F(Q(\tau,y)-2\pi\mathbf{m})=F\left(Q(\tau-\mathrm{pr}_{T}(2\pi\mathbf{m}),y-\mathrm{pr}_{Y}(2\pi\mathbf{m}))\right)=\\
=f\left(\tau-\mathrm{pr}_{T}(2\pi\mathbf{m}),y-\mathrm{pr}_{Y}(2\pi\mathbf{m})\right)
\end{gather*}
for arbitrary function $F:\mathbb{T}^{k}\mapsto\mathbb{E}^{m}$, and
then for any $\mathbf{m}\in\mathbb{M}$ and any $y\in Y_{\mathbf{m}}$
we have
\begin{gather}
y-\mathrm{pr}_{Y}(2\pi\mathbf{m})\in Y_{\mathbf{0}},\quad I_{\mathbf{\mathbf{m}}}(y)-\mathrm{pr}_{T}(2\pi\mathbf{m})=I\left(y-\mathrm{pr}_{Y}(2\pi\mathbf{m})\right),\label{eq:quat_Y-pr}
\end{gather}
Thus, for any $y\in Y_{\mathbf{m}}$, properties of function $\tau\mapsto f(\tau,y)$
on segment $I_{\mathbf{m}}(y)$ are completely defined by properties
of function $\tau\mapsto f\left(\tau,y-\mathrm{pr}_{Y}(2\pi\mathbf{m})\right)$
on segment $I\left(y-\mathrm{pr}_{Y}(2\pi\mathbf{m})\right)$ of line
$L\left(y-\mathrm{pr}_{Y}(2\pi\mathbf{m})\right)$ contained in cube
$K$.

Denote by $Y_{\ast}$ the set of such $y\in Y$ for which the function
\[
\tau\mapsto\nu(\tau,y):=\left\Vert v(\tau,y)\right\Vert ^{2}+\left\Vert v_{\tau}^{\prime}(\tau,y)\right\Vert ^{2}
\]
is not locally integrable on $\mathbb{R}$, or, what is the same thing,
for any $y\in Y_{\ast}$ there exists at least one $\mathbf{m}\in\mathbb{M}$
such that $y\in Y_{\mathbf{m}}$ and
\begin{gather}
\intop_{I_{\mathbf{m}}(y)}\nu(\tau,y)\mathrm{d}\tau=\intop_{I\left(y-\mathrm{pr}_{Y}(2\pi\mathbf{m})\right)}\nu\left(\tau,y-\mathrm{pr}_{Y}(2\pi\mathbf{m})\right)\mathrm{d}\tau=\infty.\label{eq:quat_infty}
\end{gather}
Then $Y_{\ast}\subset\bigcup_{\mathbf{m}\in\mathbb{M}}Y_{\mathbf{m}}$
and therefore $Y_{\ast}=\bigcup_{\mathbf{m}\in\mathbb{M}}(Y_{\ast}\cap Y_{\mathbf{m}})$.
Obviously,
\begin{gather*}
\mathrm{mes}(Y_{\ast}\cap Y_{\mathbf{m}})=\mathrm{mes}\left(\left[Y_{\ast}\cap Y_{\mathbf{m}}\right]-\mathrm{pr}_{Y}(2\pi\mathbf{m})\right)
\end{gather*}
and in view of~\eqref{eq:quat_Y-pr}, \eqref{eq:quat_infty} we have
$\left[Y_{\ast}\cap Y_{\mathbf{m}}\right]-\mathrm{pr}_{Y}(2\pi\mathbf{m})\subset Y_{\ast}\cap Y_{\mathbf{0}}$.
Since for almost all $y\in Y$ the function $\tau\mapsto\nu(\tau,y)$
is integrable on $I_{\mathbf{0}}(y)=I(y)$, then $\mathrm{mes}\left(Y_{\ast}\cap Y_{\mathbf{0}}\right)=0$,
and therefore
\begin{gather*}
\mathrm{mes}\left(Y_{\ast}\right)\le\sum_{\mathbf{m}\in\mathbb{M}}\mathrm{mes}(Y_{\ast}\cap Y_{\mathbf{m}})=\sum_{\mathbf{m}\in\mathbb{M}}\mathrm{mes}\left(\left[Y_{\ast}\cap Y_{\mathbf{m}}\right]-\mathrm{pr}_{Y}(2\pi\mathbf{m})\right)\le\\
\le\sum_{\mathbf{m}\in\mathbb{M}}\mathrm{mes}\left(Y_{\ast}\cap Y_{\mathbf{0}}\right)=0.
\end{gather*}
Thus, for almost all $y\in Y$ the function $\tau\mapsto\nu(\tau,y)$
is locally integrable on $\mathbb{R}$. From this it follows that
for almost all $y\in Y$ the function $\tau\mapsto v(\tau,y)$ is
absolutely continuous and has integrable derivative on any segment
$I_{\mathbf{m}}(y)$, $\mathbf{m}\in\mathbb{M}$.

Arguing in the same way as above, we can easily prove that for almost
all $y\in Y$ the function $\tau\mapsto v(\tau,y)$ is twice continuously
differentiable on each $I_{\mathbf{m}}(y)$, $\mathbf{m}\in\mathbb{M}$.
Now let us take into account that in view of~\eqref{eq:quat_1stvareq0}
integration over cube $K$ in the right-hand side of equality \eqref{eq:quat_1stvar_K}
can be replaced by integration over cube $K+s\varepsilon_{k}$ with
arbitrary $s\in\mathbb{R}$. According to equality~\eqref{eq:quat_gvwG}
the segment $I(y)$ can be replaced be $I(y)+s$. Then it becomes
clear that the function $\tau\mapsto v(\tau,y)$ is twice continuously
differentiable on each segment $I_{\mathbf{m}}(y)+s$, $\mathbf{m}\in\mathbb{M}$,
and hence, on all real axis. In addition the equality~\eqref{eq:quat_gvveqGW}
holds true for all real $\tau$. By substituting $\tau=y_{k}+\left\Vert \omega\right\Vert t$
in~\eqref{eq:quat_gvveqGW}, we arrive at conclusion that for any
$y_{k}\in\mathbb{R}$ and almost all $y\in Y$ the function $t\mapsto u\left(\sum_{i=1}^{k}y_{k}\varepsilon_{k}+t\omega\right)$,
$t\in\mathbb{R}$, is a classical solution of Lagrangian system
\begin{gather*}
\frac{\mathrm{d}}{\mathrm{d}t}\left(g(x)\dot{x}\right)=\frac{\partial}{\partial x}\left[\left(g(x)\dot{x},\dot{x}\right)+W\left(\varphi+t\omega,x\right)\right]
\end{gather*}
with $\varphi=\sum_{i=1}^{k}y_{i}\varepsilon_{i}\in K$, i.e., of
system generated on the map $\mathcal{U}$ by Lagrangian density~\eqref{eq:quat_Lagr_phi}.
It remains only to observe that the orthogonal transformation which
assigns to a point $(y_{1}\text{,}\ldots,y_{k})\in Q^{-1}(K)$ the
point $\varphi=\sum_{i=1}^{k}y_{i}\varepsilon_{i}\in K$ is measure-preserving.
\end{proof}
Let us fix a point $\varphi_{0}\in\mathbb{T}^{k}$ in such a way that
the function $t\mapsto x(t):=u(\varphi_{0}+t\omega)$ be a classical
solution from Proposition~\ref{prop:quat_classical}. If the metric
$\left\langle \cdot,\cdot\right\rangle $ were Euclidean, i.e., $g(x)\equiv\mathrm{const}$,
then the immediate consequence of Landau inequality~\cite{Lan13}
would be the boundedness of solution derivative on whole real axis.
However, in general case the substantiation of boundedness for the
derivative needs another approach.
\begin{prop}
\label{prop:quat_bound_deriv}The function $\left\Vert \dot{x}(\cdot)\right\Vert $
is bounded on real axis.\end{prop}
\begin{proof}
The vector field $\xi(t):=\dot{x}(t)$ along the curve $\gamma$
defined by equation $x=x(t)$, $t\in\mathbb{R}$, satisfies the
identity
\begin{gather*}
\nabla_{\dot{x}(t)}\xi(t)\equiv\nabla W(\varphi_{0}+t\omega,x(t))\text{.}
\end{gather*}
By applying the derivative $\nabla_{\dot{x}(t)}$ to both sides, one
ascertains that $\xi(\cdot)$ is a solution of linear inhomogeneous
system
\begin{gather*}
\nabla_{\dot{x}}^{2}\xi=H_{W}(\varphi_{0}+t\omega,x(t))\xi+\left[\nabla\frac{\partial}{\partial t}W(\varphi_{0}+t\omega,x)\right]_{x=x(t)},
\end{gather*}
and at the same time satisfies the system
\begin{gather}
\nabla_{\dot{x}}^{2}\xi=H_{W}(\varphi_{0}+t\omega,x(t))\xi-r(\left\Vert \dot{x}(t)\right\Vert )R(\dot{x}(t),\xi)\dot{x}(t)+\left[\nabla\frac{\partial}{\partial t}W(\varphi_{0}+t\omega,x)\right]_{x=x(t)},\label{eq:quat_varsysforxi}
\end{gather}
where $R$ is the curvature tensor for Levi--Civita connection, and $r(\cdot):\mathbb{R}_{+}\mapsto\mathbb{R}_{+}$
is an arbitrary continuous function (it is sufficiently to observe
that $R(\xi,\xi)=0$). If we require that $r(s)=O(s^{-2})$ for $\left|s\right|\to\infty$,
then the right-hand side of the system~\eqref{eq:quat_varsysforxi}
will be bounded with respect to $t\in\mathbb{R}$ for any field $\xi(t)$
bounded by norm.

Let $\Xi_{s}^{t}$ be an evolution operator for the linear system
$\nabla_{\dot{x}(t)}\xi=0$. By means of this operator, the parallel translation
of vectors along $\gamma$ is carried out. Namely, for any vector
$\xi_{s}\in T_{x(s)}\mathcal{M}$ the result of its parallel translation
from a point $x(s)$ to point $x(t)$ along the curve $\gamma$ is
the vector $\Xi_{s}^{t}\xi_{s}$. Under the parallel translation the
dot product for pair of vectors stays the same. Therefore the operator
$\Xi_{s}^{t}$ is orthogonal with respect to metric $\left\langle \cdot,\cdot\right\rangle $.

Set $y_{\ast}(t):=\Xi_{t}^{0}\dot{x}(t)\equiv\Xi_{t}^{0}\xi(t)$.
Since
\begin{gather*}
\nabla_{\dot{x}}\xi(t)=\lim_{s\to0}\frac{1}{s}\left[\Xi_{t+s}^{t}\xi(t+s)-\xi(t)\right]=\lim_{s\to0}\frac{1}{s}\left[\Xi_{t+s}^{t}\Xi_{0}^{t+s}y_{\ast}(t+s)-\Xi_{0}^{t}y_{\ast}(t)\right]=\Xi_{0}^{t}\dot{y}_{\ast}(t)\text{,}
\end{gather*}
then, taking into account~\eqref{eq:quat_LagrEquat}, we have
\begin{gather*}
\dot{y}_{\ast}(t)=\Xi_{t}^{0}\nabla W(\varphi_{0}+t\omega,x(t)),
\end{gather*}
and from this it follows, in particular, that
\begin{equation}
\left\Vert y_{\ast}(t)\right\Vert =O(\left|t\right|),\quad\left|t\right|\to\infty.\label{eq:quat_esty(t)}
\end{equation}
Further, since the vector field $\xi(t)$ along the curve $\gamma$
satisfies the system~\eqref{eq:quat_varsysforxi} and
\[
\nabla_{\dot{x}}^{2}\Xi_{0}^{t}y_{\ast}(t)=\nabla_{\dot{x}}\Xi_{0}^{t}\dot{y}_{\ast}(t)=\Xi_{0}^{t}\ddot{y}_{\ast}(t),
\]
 then $y_{\ast}(t)$ is a solution of linear inhomogeneous system
\begin{gather}
\ddot{y}=A(t)y+h(t)\label{eq:quat_varsysfory}
\end{gather}
where
\begin{gather*}
A(t)y:=\Xi_{t}^{0}\left[H_{W}(\varphi_{0}+t\omega,x(t))\Xi_{0}^{t}y-r(\left\Vert \xi(t)\right\Vert )R(\xi(t),\Xi_{0}^{t}y)\xi(t)\right],\\
h(t):=\Xi_{t}^{0}\left[\nabla\frac{\partial}{\partial t}W(\varphi_{0}+t\omega,x)\right]_{x=x(t)}.
\end{gather*}

Consider now the corresponding homogeneous systems
\begin{gather}
\nabla_{\dot{x}}^{2}\eta=H_{W}(\varphi_{0}+t\omega,x(t))\eta-r(\left\Vert \xi(t)\right\Vert )R(\xi(t),\eta)\xi(t)\text{,}\label{eq:quat_sysvareta}\\
\ddot{y}=A(t)y\label{eq:quat_sysvary}
\end{gather}
and show that under appropriate choice of function $r(\cdot)$ the
last ones are exponentially dichotomic. Define the function
\[
\mathcal{F}(t,\eta,\nabla_{\xi(t)}\eta):=\left\langle \nabla_{\xi(t)}\eta,\eta\right\rangle +\frac{r(\left\Vert \xi(t)\right\Vert )\left\Vert \eta\right\Vert ^{2}}{2}\left\langle \nabla V(x(t)),\xi(t)\right\rangle .
\]
By calculating its derivative along solutions
of the system~\eqref{eq:quat_sysvareta}, one can establish the following inequality:
\begin{gather*}
\frac{\mathrm{d}}{\mathrm{d}t}\mathcal{F}(t,\eta,\nabla_{\xi}\eta)
\ge\left\langle H_{W}\eta,\eta\right\rangle +\frac{r(\left\Vert \xi\right\Vert )\left\Vert \eta\right\Vert ^{2}}{2}\left\langle \nabla W,\nabla V\right\rangle +\\
+\left\Vert \nabla_{\xi}\eta\right\Vert ^{2}-r(\left\Vert \xi\right\Vert )\left\Vert \xi\right\Vert \left\Vert \nabla_{\xi}\eta\right\Vert \left\Vert \eta\right\Vert \left|\left\langle \nabla V,\varepsilon\right\rangle \right|+\frac{r(\left\Vert \xi\right\Vert )\left\Vert \xi\right\Vert ^{2}\left\Vert \eta\right\Vert ^{2}}{2}\left[\left\langle H_{V}\varepsilon,\varepsilon\right\rangle -2K^{\ast}\right]-\\
-\frac{1}{2}\left|r^{\prime}(\left\Vert \xi\right\Vert )\right|\left\Vert \xi\right\Vert \left|\left\langle \nabla W,\varepsilon\right\rangle \left\langle \nabla V,\varepsilon\right\rangle \right|\left\Vert \eta\right\Vert ^{2},
\end{gather*}
where $\varepsilon:=\xi/\left\Vert \xi\right\Vert $. Now we set
\begin{gather*}
r(s):=\begin{cases}
1, & s\in[0,B],\\
B^{2}/s^{2}, & s>B,
\end{cases}
\end{gather*}
where $B>1$, and assign $\left\Vert \nabla_{\xi}\eta\right\Vert =z_{1}$,
$\left\Vert \eta\right\Vert =z_{2}$. Then on the set of those $t$
for which $\left\Vert \xi(t)\right\Vert \le B$ we have
\begin{gather*}
\frac{\mathrm{d}}{\mathrm{d}t}\mathcal{F}(t,\eta,\nabla_{\xi(t)}\eta)\ge\left[\lambda_{W}+\frac{1}{2}\left\langle \nabla W,\nabla V\right\rangle \right]z_{2}^{2}+\\
+z_{1}^{2}-\left|\left\langle \nabla V,\varepsilon\right\rangle \right|z_{1}\left\Vert \xi\right\Vert z_{2}+\frac{\left\Vert \xi\right\Vert ^{2}z_{2}^{2}}{2}\left[\left\langle H_{V}\varepsilon,\varepsilon\right\rangle -2K^{\ast}\right]\text{.}
\end{gather*}
It is easily seen that the fulfillment of conditions of Theorem~\ref{thm:quat_1}
assures the existence of positive $\alpha_{1}$ and $\alpha_{2}$
such that
\begin{gather*}
\lambda_{W}+\frac{1}{2}\left\langle \nabla W,\nabla V\right\rangle \ge\alpha_{1}\\
z_{1}^{2}-\left|\left\langle \nabla V,\varepsilon\right\rangle \right|z_{1}\left\Vert \xi\right\Vert z_{2}+\left[\frac{1}{2}\left\langle H_{V}\varepsilon,\varepsilon\right\rangle -K^{\ast}\right]\left\Vert \xi\right\Vert ^{2}z_{2}^{2}\ge\alpha_{2}\left(z_{1}^{2}+\left\Vert \xi\right\Vert ^{2}z_{2}^{2}\right),
\end{gather*}
from what it follows that
\begin{gather*}
\frac{\mathrm{d}}{\mathrm{d}t}\mathcal{F}(t,\eta,\nabla_{\xi}\eta)\ge\alpha_{1}\left\Vert \eta\right\Vert ^{2}+\alpha_{2}\left\Vert \nabla_{\xi}\eta\right\Vert ^{2}
\end{gather*}
if $\left\Vert \xi(t)\right\Vert \le B$.

Set
\begin{gather*}
C:=\max\left\{ \max_{x\in\bar{\Omega}}\left\Vert \nabla V(x)\right\Vert ,\max_{(\varphi,x)\in\mathbb{T}^{k}\times\bar{\Omega}}\left\Vert \nabla W(\varphi,x)\right\Vert ,\max_{(\varphi,x)\in\mathbb{T}^{k}\times\bar{\Omega}}\left\Vert H_{W}(\varphi,x)\right\Vert \right\}
\end{gather*}
and choose $B$ so large that $\alpha_{2}B^{2}\ge1+C(1+3C/2)$. Then
on the set of those $t$ for which $\left\Vert \xi(t)\right\Vert >B$
we get
\begin{gather*}
\frac{\mathrm{d}}{\mathrm{d}t}\mathcal{F}(t,\eta,\nabla_{\xi}\eta)\ge\left\Vert \nabla_{\xi}\eta\right\Vert ^{2}-B\left|\left\langle \nabla V,\varepsilon\right\rangle \right|\left\Vert \nabla_{\xi}\eta\right\Vert \left\Vert \eta\right\Vert +\frac{B^{2}\left\Vert \eta\right\Vert ^{2}}{2}\left[\left\langle H_{W}\varepsilon,\varepsilon\right\rangle -2K^{\ast}\right]-\\
-\left[C+3C^{2}/2\right]\left\Vert \eta\right\Vert ^{2}\ge\alpha_{2}(z_{1}^{2}+B^{2}z_{2}^{2})-\left[C+3C^{2}/2\right]z_{2}^{2}\ge\alpha_{2}z_{1}^{2}+z_{2}^{2}.
\end{gather*}
From this it follows that there exists $\alpha>0$ such that the derivative
of quadratic form of variables $y$, $\dot{y}$
\[
\mathcal{F}(t,\Xi_{0}^{t}y,\Xi_{0}^{t}\dot{y})=\left\langle \dot{y},y\right\rangle +\frac{r(\left\Vert \xi(t)\right\Vert )}{2}\left\langle \nabla V(x(t)),\xi(t)\right\rangle \left\Vert y\right\Vert ^{2}
\]
 along solutions of the system~\eqref{eq:quat_sysvary} satisfies
the inequalities
\begin{gather*}
\frac{\mathrm{d}}{\mathrm{d}t}\mathcal{F}(t,\Xi_{0}^{t}y,\Xi_{0}^{t}\dot{y})\ge\alpha\left[\left\Vert \Xi_{0}^{t}y\right\Vert ^{2}+\left\Vert \Xi_{0}^{t}\dot{y}\right\Vert ^{2}\right]=\alpha\left[\left\Vert y\right\Vert ^{2}+\left\Vert \dot{y}\right\Vert ^{2}\right]
\end{gather*}
 and hence is positive definite. At the same time the form $\mathcal{F}(t,\Xi_{0}^{t}y,\Xi_{0}^{t}\dot{y})$
is, obviously, nondegenerate and has bounded coefficients. As it follows
from~\cite{Sam01}, the existence of quadratic form with such properties
guarantees the exponential dichotomy of liner system~\eqref{eq:quat_sysvary}
on the whole real axis, and the inhomogeneous system~\eqref{eq:quat_varsysfory}
in view of boundedness of $\left\Vert h(t)\right\Vert $ possesses
a unique bounded solution. Furthermore, any other solution of this
system exponentially increases either for $t\to\infty$, or for $t\to-\infty$.
Since above we have established the estimate~\eqref{eq:quat_esty(t)},
then the solution $y_{\ast}(t)$ is bounded on $\mathbb{R}$. Thus
the function $\left\Vert \dot{x}(t)\right\Vert =\left\Vert \Xi_{0}^{t}y_{\ast}(t)\right\Vert =\left\Vert y_{\ast}(t)\right\Vert $
is bounded on~$\mathbb{R}$.\end{proof}
\begin{prop}
\label{prop:quat_uniqBS} If the conditions of Theorem~\ref{thm:quat_1}
holds true, then for any $\varphi\in\mathbb{T}^{k}$ the system with
Lagrangian density \eqref{eq:quat_Lagr_phi} cannot have more than
one solution $x(\cdot):\mathbb{R}\mapsto\mathcal{K}$ such that $\sup_{t\in\mathbb{R}}\left\Vert \dot{x}(t)\right\Vert <\infty$,
where $\mathcal{K}\subset\Omega$ is the compact set from Proposition~\ref{prop:quat_classical}. \end{prop}
\begin{proof}
Use reductio ad absurdum: suppose that there exists a pair of different
solutions $x_{i}(\cdot):\mathbb{R}\mapsto\mathcal{K}$ such that $\sup_{t\in\mathbb{R}}\left\Vert \dot{x}_{i}(t)\right\Vert <\infty$,
$i=1,2$. Making use of connecting mapping $\chi$ (see the end of
Sect.2) and introducing the notations
\begin{gather*}
\eta(s,t):=\frac{\partial}{\partial t}\chi\left(s,x_{1}(t),x_{2}(t)\right),\quad\xi(s,t):=\frac{\partial}{\partial s}\chi\left(s,x_{1}(t),x_{2}(t)\right),
\end{gather*}
define the function
\begin{gather*}
l(t):=\left\langle \eta(s,t),\xi(s,t)\right\rangle \Bigl|_{s=0}^{s=1}=\\
=\left\langle \dot{x}_{2}(t),\xi(1,t)\right\rangle -\left\langle \dot{x}_{1}(t),\xi(0,t)\right\rangle
\end{gather*}
This function is bounded on $\mathbb{R}$, and it turns out that its derivative can
be represented in the form
\begin{gather*}
\dot{l}(t)
=\frac{\partial}{\partial s}L\left(\varphi+t\omega,\chi\left(s,x_{1}(t),x_{2}(t)\right),\eta(s,t)\right)\Bigl|_{s=0}^{s=1}=\\
=\intop_{0}^{1}\frac{\partial^{2}}{\partial s^{2}}L\left(\varphi+t\omega,\chi\left(s,x_{1}(t),x_{2}(t)\right),\eta(s,t)\right)\mathrm{d}s.
\end{gather*}
Taking into account~\eqref{eq:quat_convexL} we have
\begin{gather*}
\dot{l}(t)\ge\varkappa\intop_{0}^{1}\left[\left\Vert \nabla_{\xi}\eta\right\Vert ^{2}+\left(\left\Vert \eta\right\Vert ^{2}+1\right)\left\Vert \xi\right\Vert ^{2}\right]\mathrm{d}s.
\end{gather*}
Thus, the function $l(\cdot)$ is nondecreasing and its boundedness
assures convergence of the integrals
\begin{gather}
\intop_{-\infty}^{0}\intop_{0}^{1}\left\Vert \xi(s,t)\right\Vert ^{2}\mathrm{d}s\mathrm{d}t<\infty,\quad\intop_{0}^{\infty}\intop_{0}^{1}\left\Vert \xi(s,t)\right\Vert ^{2}\mathrm{d}s\mathrm{d}t<\infty.\label{eq:quat_converg_int}
\end{gather}

Note that $\chi_{s}^{\prime}(0,x,y)\ne0$ once $x\ne y$. In fact,
otherwise the mapping $s\mapsto\chi(s,x,y)$ would be a solution of
equation~\eqref{eq:quat_prop_chi} which satisfies the initial conditions
$\chi\bigl|_{s=0}=x$, $\chi_{s}^{\prime}\bigl|_{s=0}=0$, but only
constant solution $\chi(s,x,y)\equiv x$ has such a property. This
contradicts the equality $\chi(1,s,y)=y$. Taking into account that
the equality $x_{1}(t)=x_{2}(t)$ can be valid only on a discrete
set, we have $\xi(0,t)\not\equiv0$. But then there hold  the strict
inequalities
\begin{gather}
\limsup_{t\to-\infty}l(t)<l(0)<\liminf_{t\to\infty}l(t).\label{eq:quat_liminfsup_l}
\end{gather}

It is not hard to show that the second inequality~\eqref{eq:quat_converg_int}
assures the existence of sequence $t_{k}^{+}\to\infty$ such that
$\max\left\{ \left\Vert \xi(0,t_{k}^{+})\right\Vert ,\left\Vert \xi(1,t_{k}^{+})\right\Vert \right\} \to0$.
In the same way, one can prove that the first inequality in~\eqref{eq:quat_converg_int}
guarantees the existence of sequence $t_{k}^{-}\to-\infty$ along
which $\left\Vert \xi(0,t)\right\Vert $ and $\left\Vert \xi(1,t)\right\Vert $
simultaneously tend to zero. Then taking into account boundedness
of $\left\Vert \dot{x}_{i}(t)\right\Vert $ on $\mathbb{R}$, we get
\begin{gather*}
\lim_{k\to\infty}l(t_{k}^{\pm})=0.
\end{gather*}
This equality contradicts~\eqref{eq:quat_liminfsup_l}.
Namely, the case $l(0)\ge0$ (the case $l(0)<0$) contradicts the
existence of sequence $\left\{ t_{k}^{+}\right\} $ (the existence
of sequence $\left\{ t_{k}^{-}\right\} $).
\end{proof}
\emph{Proof of Theorem~\ref{thm:quat_class_QPS}. }Apply the Amerio
theorem (see, e.g., \cite[p. 437]{Dem67}). On the map $\mathcal{U}\subset\mathbb{E}^{m}$
of domain $\Omega\subset\mathcal{M}$, the system with Lagrangian density
$L(\varphi_{0}+t\omega,x,\dot{x})$ takes the form of a second order
system, which is equivalent to an $2m$-dimensional normal first order
system on phase space $\mathcal{U}\times\mathbb{E}^{m}$. Since components
of frequency vector $\omega$ are rationally independent, then the
so-called $H$-class of system with Lagrangian density $L(\varphi_{0}+t\omega,x,\dot{x})$
is formed by a family of systems with Lagrangian densities $L(\varphi+t\omega,x,\dot{x})$
parametrized by points of torus $\varphi\in\mathbb{T}^{k}$. From
Propositions~\ref{prop:quat_classical}--\ref{prop:quat_uniqBS}
it follows that in $\mathcal{C}\times\mathbb{E}^{m}$ ($\mathcal{C}$
is the image of compact set $\mathcal{K}$ on the map) the system
with Lagrangian density $L(\varphi_{0}+t\omega,x,\dot{x})$ has a
unique classical bounded Besicovitch quasiperiodic solution $t\mapsto u(\varphi_{0}+t\omega)$
 and each system from its $H$-class has a unique bounded solution
in $\mathcal{C}\times\mathbb{E}^{m}.$ Thus, all requirements of the
Amerio theorem are fulfilled and $u(\varphi_{0}+t\omega)$ is a uniformly
almost periodic function. In view of its Fourier series, this function
is uniformly quasiperiodic  with frequency basis $\omega$.
Then $u(\cdot)\in\mathrm{C}\!\left(\mathbb{T}^{k}\!\mapsto\!\mathcal{C}\right)$
and for any $\varphi\in\mathbb{T}^{k}$, in particular for $\varphi=0$,
the function $t\mapsto u(\varphi+t\omega)$ is a classical quasiperiodic
solution of the system with Lagrangian density $L(\varphi+t\omega,x,\dot{x})$.
\begin{thm}
\label{thm:quat_ExpDych}Suppose that the conditions of Theorem~\ref{thm:quat_class_QPS}
holds true and let $x(\cdot)\in\mathrm{C}^{2}\!\left(\mathbb{R}\!\mapsto\!\mathcal{K}\right)$
be a quasiperiodic solution of system with Lagrangian density $L(t\omega,x.\dot{x})$.
Then the system in variations along this solution is exponentially
dichotomic on $\mathbb{R}$. \end{thm}
\begin{proof}
Let $t\mapsto x(t,s)$, $s\in(-\delta,\delta)$, be a family of solutions
to system with Lagrangian density $L(t\omega,x,\dot{x})$. Let this family smoothly
depends on parameter $s$ and $x(t,0)=x(t)$. Then $\xi(t,s):=x_{t}^{\prime}(t,s)$,
$\eta(t,s):=x_{s}^{\prime}(t,s)$ define vector fields along the mapping
$(t,s)\mapsto x(t,s)$ and there holds the Lagrange equation
\begin{gather*}
\nabla_{\xi(t,s)}\xi(t,s)=\nabla W(t\omega+\varphi,x(t,s))\text{.}
\end{gather*}
Calculate here the covariant derivative $\nabla_{\eta(t,s)}$ from
both sides, take into account the equalities~\cite[Sect. 3.6]{GKM71}
\begin{gather*}
\nabla_{\eta}\xi=\nabla_{\xi}\eta,\quad\nabla_{\eta}\nabla_{\xi}\xi-\nabla_{\xi}\nabla_{\eta}\xi=R(\xi,\eta)\xi\text{,}
\end{gather*}
 where$R$ is curvature tensor f Levi--Civita connection, and set
$s=0,$ $\xi(t):=\xi(t,0)$, $\eta(t):=\eta(t,0)$. Then we arrive
at conclusion that the field $\eta(t)$ satisfies the system in variations
along $x(t)$:
\begin{gather}
\nabla_{\xi(t)}^{2}\eta=\nabla_{\eta}\nabla W(t\omega+\varphi,x(t))-R(\xi(t),\eta)\xi(t).\label{eq:quat_sys_var}
\end{gather}
Since in this case the function $t\mapsto\left\Vert \xi(t)\right\Vert $
is bounded, then we can make use of reasonings from the proof of Proposition~\ref{prop:quat_bound_deriv}
with $r(s)\equiv1$ to show the exponential dichotomy of system~\eqref{eq:quat_sys_var}
\end{proof}
\textbf{4. Concluding remarks. } The authors
 of papers~\cite{BerZha96,BloPen01}  restrict themselves to proof of the fact that classical solutions
generated by weak solutions of certain classes of systems on Euclidean
space are Besicovitch quasiperiodic functions, while by our main result we shows that under conditions of Theorem~\ref{thm:quat_class_QPS} such weak solutions are actually  classical uniformly quasiperiodic ones.  It is also worth to note
that the paper \cite{BloPen01} deals with quasiperiodic systems in
$\mathbb{E}^{m}$ of rather general form
\begin{gather*}
\frac{\mathrm{d}^{p}q}{\mathrm{d}t^{p}}=F\left(t\omega,q,\dot{q},\ldots,\frac{\mathrm{d}^{p-1}q}{\mathrm{d}t^{p-1}}\right),
\end{gather*}
but under boundedness condition of right-hand side with respect to derivatives
$q^{(j)}$, $j=1,\ldots,p-1$. Unfortunately, this condition makes
impossible application of results from \cite{BloPen01}, concerning
classical Besicovitch quasiperiodic solutions, to Lagrangian systems
on Riemannian manifold with nonconstant metric tensor $g(x)$.

Application of our results to concrete mechanical systems requires
constructing the auxiliary function $V(\cdot)$. One of the ways to
find such a function by means of averaged force function
\begin{gather*}
\bar{W}(x):=\frac{1}{(2\pi)^{k}}\intop_{\mathbb{T}^{k}}W(\varphi,x)\mathrm{d}\varphi
\end{gather*}
has been offered in \cite{ParRus12}.

\end{document}